\documentclass{llncs}
\let\citep\cite

\usepackage{mathtools}
\usepackage{amssymb}
\usepackage{shuffle}
\DeclareMathOperator{\sgn}{sgn}

\DeclareMathOperator{\ins}{Ins}
\DeclareMathOperator{\alphabet}{alph}

\usepackage{tikz}
\usetikzlibrary{automata,arrows}
\tikzset{state/.style={draw,circle,minimum width=6pt}}
\tikzset{every edge/.style={draw,->,>=stealth',shorten >=1pt,auto,semithick}}
\tikzset{initial text={},double distance=2pt}
\usetikzlibrary{calc}
\usetikzlibrary{decorations.pathreplacing}

\usepackage[hidelinks]{hyperref}
\usepackage{color}

\title{Algorithms and Training for Weighted Multiset Automata and Regular Expressions}
\titlerunning{Multiset Automata}

\author{Justin DeBenedetto \and David Chiang}
\institute{Department of Computer Science and Engineering \\
University of Notre Dame, Notre Dame, IN 46556, USA\\
\email{\tt \{jdebened,dchiang\}@nd.edu}}
\authorrunning{DeBenedetto and Chiang}

\begin{document}

\maketitle

\begin{abstract}
Multiset automata are a class of automata for which the symbols can be read in any order and obtain the same result.  We investigate weighted multiset automata and show how to construct them from weighted regular expressions.  We present training methods to learn the weights for weighted regular expressions and for general multiset automata from data.  Finally, we examine situations in which inside weights can be computed more efficiently.
\end{abstract}

\section{Introduction}

Automata have been widely studied and utilized for pattern and string matching problems. A string automaton reads the symbols of an input string one at a time, after which it accepts or rejects the string. But in certain instances, the order in which the symbols appear is irrelevant.

For example, in a set of graphs, the edges incident to a node are unordered and therefore their labels form a commutative language. Or, in natural language processing, applications might arise in situations where a sentence is generated by a context-free grammar subject to (hard or soft) order-independent constraints. For example, in summarization, there might be an unordered set of facts that must be included. Or, there might be a constraint that among the references to a particular entity, exactly one is a full NP. 

To handle these scenarios, we are interested in weighted automata and weighted regular expressions for multisets. 
This paper makes three main contributions:
\begin{itemize}
\item We define a new translation from weighted multiset regular expressions to weighted multiset automata, more direct than that of Chiang et al.~\cite{chiang+:cl2018} and more compact (but less general) than that of Droste and Gastin~\cite{droste+gastin:1999}.
\item We discuss how to train weighted multiset automata and regular expressions from data.
\item We give a new composable representation of partial runs of weighted multiset automata that is more efficient than that of Chiang et al.~\cite{chiang+:cl2018}.
\end{itemize}

\section{Definitions}

We begin by defining weighted multiset automata (\S\ref{sec:weightedMultisetAutomata}) and the related definitions from previous papers for weighted multiset regular expressions (\S\ref{sec:weightedMultisetRegularExpressions}).

\subsection{Preliminaries}

For any natural number $n$, let $[n] = \{1, \ldots, n\}$.

A \emph{multiset} over a finite alphabet $\Sigma$ is a mapping from $\Sigma$ to $\mathbb{N}_0$. For consistency with standard notation for strings, we write $a$ (where $a \in \Sigma$) instead of $\{a\}$, $uv$ for the multiset union of multisets $u$ and $v$, and $\epsilon$ for the empty multiset.

The \emph{Kronecker product} of a $m \times n$ matrix $A$ and a $p \times q$ matrix $B$ is the $mp \times nq$ matrix
\[ A \otimes B = \begin{bmatrix}
A_{11} B & \cdots & A_{1m} B \\
\vdots & \ddots &\vdots \\
A_{n1} B & \cdots & A_{mn} B \\
\end{bmatrix}.
\]

If $w$ is a string over $\Sigma$, we write $\alphabet(w)$ for the subset of symbols actually used in $w$; similarly for $\alphabet(L)$ where $L$ is a language. If $|\alphabet(L)|=1$, we say that $L$ is \emph{unary}.

\subsection{Weighted multiset automata}
\label{sec:weightedMultisetAutomata}

We formulate weighted automata in terms of matrices as follows.
Let $\mathbb{K}$ be a commutative semiring.

\begin{definition}
A \emph{$\mathbb{K}$-weighted finite automaton} (WFA) over $\Sigma$ is a tuple $M=(Q, \Sigma, \lambda, \mu, \rho)$, where
$Q = [d]$ is a finite set of states,
$\Sigma$ is a finite alphabet,
$\lambda \in \mathbb{K}^{1 \times d}$ is a row vector of initial weights,
$\mu : \Sigma \rightarrow \mathbb{K}^{d \times d}$ assigns a transition matrix to every symbol, and
$\rho \in \mathbb{K}^{d \times 1}$ is a column vector of final weights.
\end{definition}
For brevity, we extend $\mu$ to strings: If $w \in \Sigma^\ast$, then $\mu(w) = \mu(w_1) \cdots \mu(w_n)$. 
Then, the weight of all paths accepting $w$ is
$M(w) = \lambda \, \mu(w) \, \rho$.
Note that in this paper we do not consider $\epsilon$-transitions.
Note also that one unusual feature of our definition is that it allows a WFA to have more than one initial state.

\begin{definition}
A \emph{$\mathbb{K}$-weighted multiset finite automaton} is one whose transition matrices commute pairwise. That is, for all $a, b \in \Sigma$, we have $\mu(a)\mu(b) = \mu(b)\mu(a)$.
\end{definition}

\subsection{Weighted multiset regular expressions}
\label{sec:weightedMultisetRegularExpressions}

This definition follows that of Chiang et al.~\cite{chiang+:cl2018}, which in turn is a special case of that of Droste and Gastin~\cite{droste+gastin:1999}.
\begin{definition}
A \emph{$\mathbb{K}$-weighted multiset regular expression} over $\Sigma$ is an expression belonging to the smallest set $\mathcal{R}(\Sigma)$ satisfying:
\begin{itemize}
\item If $a \in \Sigma$, then $a \in \mathcal{R}(\Sigma)$.
\item $\epsilon \in \mathcal{R}(\Sigma)$.
\item $\emptyset \in \mathcal{R}(\Sigma)$.
\item If $\alpha, \beta \in \mathcal{R}(\Sigma)$, then $\alpha \cup \beta \in \mathcal{R}(\Sigma)$.
\item If $\alpha, \beta \in \mathcal{R}(\Sigma)$, then $\alpha \beta \in \mathcal{R}(\Sigma)$.
\item If $\alpha \in \mathcal{R}(\Sigma)$, then $\alpha^\ast \in \mathcal{R}(\Sigma)$.
\item If $\alpha \in \mathcal{R}(\Sigma)$ and $k \in \mathbb{K}$, then $k\alpha \in \mathcal{R}(\Sigma)$.
\end{itemize}
\end{definition}
We define the language described by a regular expression, $\mathcal{L}(\alpha)$, by analogy with string regular expressions. Note that $\epsilon$ matches the empty multiset, while $\emptyset$ does not match any multisets. Interspersing weights in regular expressions allows regular expressions to describe weighted languages.
\begin{definition}
A multiset mc-regular expression is one where in every subexpression $\alpha^\ast$, $\alpha$ is:
\begin{itemize}
\item proper: $\epsilon \notin \mathcal{L}(\alpha)$, and
\item monoalphabetic and connected: $\mathcal{L}(\alpha)$ is unary.
\end{itemize}
\end{definition}
As an example of why these restrictions are needed, consider the regular expression $(ab)^\ast$. Since the symbols commute, this is equivalent to $\{a^nb^n\}$, which multiset automata would not be able to recognize.
From now on, we assume that all multiset regular expressions are mc-regular and do not write ``mc-.''

\section{Matching Regular Expressions}

In this section, we consider the problem of computing the weight that a multiset regular expression assigns to a multiset. The bad news is that this problem is NP-complete (\S\ref{sec:npcomplete}). However, we can convert a multiset regular expression to a multiset automaton (\S\ref{sec:RegExpToAutomata}) and run the automaton.

\subsection{NP-completeness}
\label{sec:npcomplete}

\begin{theorem}
The membership problem for multiset regular expressions is NP-complete. \end{theorem}
\begin{proof}
Define a transformation $\mathcal{T}$ from Boolean formulas in CNF over a set of variables $X$ to multiset regular expressions over the alphabet $X \cup \{\bar{x} \mid x \in X\}$:
\begin{align*}
\mathcal{T}(\phi_1 \lor \phi_2) &= \mathcal{T}(\phi_1) \cup \mathcal{T}(\phi_2) \\
\mathcal{T}(\phi_1 \land \phi_2) &= \mathcal{T}(\phi_1) \mathcal{T}(\phi_2) \\
\mathcal{T}(x) &= x \\
\mathcal{T}(\neg x) &= \bar{x}
\end{align*}
Given a formula $\phi$ in 3CNF, construct the multiset regular expression $\alpha = \mathcal{T}(\phi)$. Let $n$ be the number of clauses in $\phi$.
Then form the expression
\[ \beta = \prod_x \left(x^n (\bar{x} \cup \epsilon)^n \cup (x \cup \epsilon)^n \bar{x}^n\right) \]
Both $\alpha$ and $\beta$ clearly have length linear in $n$. We claim that $\phi$ is satisfiable if and only if $L(\alpha\beta)$ contains $w = \prod_x x^n \bar{x}^n$. 

\begin{trivlist}
\item $(\Rightarrow)$ If $\phi$ is satisfiable, form a string $u = u_1 \cdots u_n$ as follows. For $i = 1, \ldots n$, the $i$th clause of $\phi$ has at least one literal made true by the satisfying assignment. If it's $x$, then $u_i = x$; if it's $\neg x$, then $u_i = \bar{x}$. Clearly, $u \in L(\alpha)$. Next, form a string $v = \prod_x v_x$, where the $v_x$ are defined as follows. For each $x$, if $x$ is true under the assignment, then there are $k \geq 0$ occurrences of $x$ in $u$ and zero occurrences of $\bar{x}$ in $u$. Let $v_x = x^{n-k} \bar{x}^n$. Likewise, if $x$ is false under the assignment, then there are $k \geq 0$ occurrences of $\bar{x}$ and zero occurrences of $x$, so let $v_x = x^k \bar{x}^{n-k}$. Clearly, $uv = w$ and $v \in L(\beta)$.
\medskip
\item $(\Leftarrow)$ If $w \in L(\alpha\beta)$, then there exist strings $uv=w$ such that $u \in L(\alpha)$ and $v \in L(\beta)$. For each $x$, it must be the case that $v$ contains either $x^n$ or $\bar{x}^n$, so that $u$ must either not contain $x$ or not contain $\bar{x}$. In the former case, let $x$ be false; in the latter case, let $x$ be true. The result is a satisfying assignment for~$\phi$. \qed
\end{trivlist}
\end{proof}

\subsection{Conversion to multiset automata}
\label{sec:RegExpToAutomata}
Given a regular expression $\alpha$, we can construct a finite multiset automaton corresponding to that regular expression. 
In addition to $\lambda$, $\mu(a)$, and $\rho$, we compute Boolean matrices $\kappa(a)$ with the same dimensions as $\mu(a)$. The interpretation of these matrices is that whenever the automaton is in state $q$, then $[\kappa(a)]_{qq} = 1$ iff the automaton has not read an $a$ yet.

If $\alpha = a$, then for all $b \neq a$:
\begin{align*}
\lambda &= \begin{bmatrix}1 & 0 \end{bmatrix} &
\mu(a) &= \begin{bmatrix}
0 & 1 \\
0 & 0 \\
\end{bmatrix} &
\kappa(a) &= \begin{bmatrix}1 & 0 \\ 0 & 0 \end{bmatrix} &
\rho &= \begin{bmatrix}0 \\ 1 \end{bmatrix}
\\
&& \mu(b) &= \begin{bmatrix}
0 & 0 \\
0 & 0 \\
\end{bmatrix} & \kappa(b) &= \begin{bmatrix}1 & 0 \\0 & 1\end{bmatrix}.
\end{align*}
If $\alpha = k \alpha_1$ (where $k \in \mathbb{K}$), then for all $a \in \Sigma$:
\begin{align*}
\mu(a)&=\mu_1(a) &
\lambda &= \lambda_1 &
\rho &= k \rho_1 &
\kappa(a) &= \kappa(a).
\end{align*}
If $\alpha = \alpha_1 \cup \alpha_2$, then for all $a \in \Sigma:$
\begin{align*}
\mu(a) = \begin{bmatrix}\mu_1(a) & 0 \\ 0 & \mu_2(a)\end{bmatrix}
\lambda &= \begin{bmatrix}\lambda_1 & \lambda_2 \end{bmatrix} &
\rho &= \begin{bmatrix}\rho_1 \\ \rho_2 \end{bmatrix} &
\kappa(a) &= \begin{bmatrix}\kappa_1(a) & 0 \\ 0 & \kappa_2(a) \end{bmatrix}
\end{align*}
If $\alpha = \alpha_1 \alpha_2$, then for all $a \in \Sigma$:
\begin{align*}
\mu(a) &= \mu_1(a) \otimes \kappa_2(a) + I \otimes \mu_2(a) &
\lambda &= \lambda_1 \otimes \lambda_2 \\
\kappa(a) &= \kappa_1(a) \otimes \kappa_2(a) &
\rho &= \rho_1 \otimes \rho_2.
\end{align*}
If $\alpha = \alpha_1^\ast$ and $\alpha_1$ is unary, then for all $a \in \Sigma:$
\begin{align*}
\mu(a) &= \mu_1(a) + \rho \lambda \mu_1(a) &
\lambda &= \lambda_1 &
\rho &= \rho_1 + \lambda_1^{\top} &
\kappa(a) &= \kappa_1(a).
\end{align*}

This construction can be explained intuitively as follows. The case $\alpha=a$ is standard.  
The union operation is standard except that the use of two initial states makes for a simpler formulation.  
The shuffle product is similar to a conventional shuffle product except for the use of $\kappa_2$. It builds an automaton whose states are pairs of states of the automata for $\alpha_1$ and $\alpha_2$. The first term in the definition of $\mu(a)$ feeds $a$ to the first automaton and the second term to the second; but it can be fed to the first only if the second has not already read an $a$, as ensured by $\kappa_2(a)$. Finally, Kleene star adds a transition from final states to ``second'' states (states that are reachable from the initial state by a single $a$-transition), while also changing all initial states into final states.

Let $A(\alpha)$ denote the multiset automaton constructed from $\alpha$. We can bound the number of states of $A(\alpha)$ by $2^{|\alpha|}$ by induction on the structure of $\alpha$.  For $\alpha = \epsilon$, $|A(\alpha)| = 1 \leq 2^{|\alpha|}$.  For $\alpha = a$, $|A(\alpha)| = 2 \leq 2^{|\alpha|}$.  For $\alpha = \alpha_1 \bigcup \alpha_2$, $|A(\alpha)| = |A(\alpha_1)| + |A(\alpha_2)| \leq 2^{|\alpha|}$.  For $\alpha = \alpha_1\alpha_2$, $|A(\alpha)| = |A(\alpha_1)| |A(\alpha_2)| \leq 2^{|\alpha|}$.  For $\alpha = \alpha_1^\ast$, $|A(\alpha)| = |A(\alpha_1)| \leq 2^{|\alpha|}$.

\subsection{Related work}

Droste and Gastin~\cite{droste+gastin:1999} show how to perform regular operations for the more general case of trace automata (automata on monoids). 
Our use of $\kappa$ resembles their forward alphabet.  
Our construction does not utilize anything akin to their backward alphabet, so that we allow outgoing edges from final states and we allow initial states to be final states.  
Their construction, when converting $\alpha_1^\ast$, creates $m=|\alphabet(\alpha_1)|$ simultaneous copies of $A(\alpha_1)$, that is, it creates an automaton with $|A(\alpha_1)|^m$ states. Since our Kleene star is restricted to the unary case, we can use the standard, much simpler, Kleene star construction~\cite{berry1986regular}. 

Our construction is a modification of a construction from previous work~\cite{chiang+:cl2018}. 
Previously, the shuffle operation required $\alphabet(\alpha_1)$ and $\alphabet(\alpha_2)$ to be disjoint; to ensure this required some rearranging of the regular expression before converting to an automaton.   
Our construction, while sharing the same upper bound on the number of states, operates directly on the regular expression without any preprocessing.

\section{Learning Weights}

Given a collection of multisets, the weights of the transition matrices and the initial and final weights can be learned automatically from data.  
Given a multiset $w,$ we let $\mu(w)= \prod_{i} \mu(w_i)$.  The probability of $w$ over all possible multisets 
is 
\begin{align*}
P(w) &= \frac1Z \lambda \mu(w) \rho \\
Z &= \sum_{\text{multisets $w'$}} \lambda \mu(w') \rho.
\end{align*}
We must restrict $w'$ to multisets up to a given length bound, which can be set based on the size of the largest multiset which is reasonable to occur in the particular setting of use.  
Without this restriction, the infinite sum for $Z$ 
will diverge in many cases.  For example, if $\alpha = a^\ast$, then $\mu(a)^n=\mu(a)$ and thus $\lambda \mu(a) \rho = \lambda \mu(a)^n \rho$. Since this value is non-zero, the sum diverges.

The goal is to minimize the negative log-likelihood given by 
\[ L = - \sum_{w \in \text{data}} \log P(w). \]
To this end, we envision and describe two unique scenarios for how the multiset automata are formed.  

\subsection{Regular expressions}
In certain circumstances, we may start with a set of rules as weighted regular expressions and wish to learn the weights from data.  Conversion from weighted regular expressions to multiset automata can be done automatically, see Section~\ref{sec:RegExpToAutomata}. Now the multiset automata that result already have commuting transition matrices.  The weights from the weighted regular expression are the parameters to be learned.  
These parameters can be learned through stochastic gradient descent with the gradient computed through automatic differentiation, and the transition matrices will retain their commutativity by design.

\subsection{Finite automata}
We can learn the weighted automaton entirely from data by starting with a fully connected automaton on $n$ nodes.  All initial, transition, and final weights are initialized randomly.  
Learning proceeds by gradient descent on the log-likelihood with a penalty to encourage the transition matrices to commute.  Thus our modified log-likelihood is  
$$L' = L + \alpha \sum_{a,b} (\mu(a) \mu(b) - \mu(b)\mu(a))$$
Over time we increase the penalty by increasing $\alpha$.  This method has the benefit of allowing us to learn the entire structure of the automaton directly from data without having to form rules as regular expressions.  Additionally, since we set $n$ at the start, the number of states can be kept small and computationally feasible.  The main drawback of this method is that the transition matrices, while penalized for not commuting, may not exactly satisfy the commuting condition.

\section{Computing Inside Weights}

We can compute the total weight of a multiset incrementally by starting with $\lambda$ and multiplying by $\mu(a)$ for each $a$ in the multiset. But in some situations, we might need to compose the weights of two partial runs. That is, having computed $\mu(u)$ and $\mu(v)$, we want to compute $\mu(uv)$ in the most efficient way. Sometimes we also want to be able to compute $\mu(u)+\mu(v)$ in the most efficient way.

For example, if we divide $w$ into parts $u$ and $v$ to compute $\mu(u)$ and $\mu(v)$ in parallel \cite{ladner1980parallel}, afterwards we need to compose them to form $\mu(w)$. Or, we could intersect a context-free grammar with a multiset automaton, and parsing with the CKY algorithm would involve multiplying and adding these weight matrices. The recognition algorithm for extended DAG automata \cite{chiang+:cl2018} uses multiset automata in this way as well.

Let $M$ be a multiset automaton and $\mu(a)$ its transition matrices. Let us call $\mu(w)$ the matrix of \emph{inside weights} of $w$.  If stored in the obvious way, it takes $\mathcal{O}(d^2)$ space. If $w=uv$ and we know $\mu(u)$ and $\mu(v)$, we can compute $\mu(w)$ by matrix multiplication in~$\mathcal{O}(d^3)$ time. Can we do better?

The set of all matrices $\mu(w)$ spans a 
module
which we call $\ins(M)$. We show in this section that, under the right conditions, if $M$ has $d$ states, then $\ins(M)$ has 
a generating set of size 
$d$, so that we can represent $\mu(w)$ as a vector of $d$ coefficients. We begin with the special case of unary languages (\S\ref{sec:unary}), then after a brief digression to more general languages (\S\ref{sec:binary}), we consider multiset regular expressions converted to multiset automata (\S\ref{sec:combination}).

\subsection{Unary languages}
\label{sec:unary}

Suppose that the automaton is unary, that is, over the alphabet $\Sigma = \{a\}$. Throughout this section, we write $\mu$ for $\mu(a)$ for brevity.

\subsubsection{Ring-weighted}

The inside weights of a string $w = a^n$ are simply the matrix $\mu^n$, and the inside weights of a set of strings is a polynomial in $\mu$. We can take this polynomial to be our representation of inside weights, if we can limit the degree of the polynomial.

The Cayley-Hamilton theorem (CHT) says that any matrix $\mu$ over a commutative ring satisfies its own \emph{characteristic equation},
$
\det (\lambda I-\mu) = 0,
$
by substituting $\mu$ for $\lambda$. The left-hand side of this equation is the \emph{characteristic polynomial}; its highest-degree term is $\lambda^d$. So if we substitute $\mu$ into 
the characteristic equation
and solve for $\mu^d$, we have a way of rewriting any polynomial in $\mu$ of degree $d$ or more into a polynomial of degree less than $d$.

So representing the inside weights as a polynomial in $\mu$ takes only $O(d)$ space, and addition takes $O(d)$ time.
Naive multiplication of polynomials takes $O(d^2)$ time; fast Fourier transform can be used to speed this up to $O(d \log d)$ time, although $d$ would have to be quite large to make this practical.

\subsubsection{Semiring-weighted}

Some very commonly used weights do not form rings: for example, the Boolean semiring, used for unweighted automata, and the Viterbi semiring, used to find the highest-weight path for a string.

There is a version of CHT for semirings due to Rutherford~\cite{rutherford:1964}.  In a ring, the characteristic equation can be expressed using the sums of determinants of principal minors of order $r$.  Denote the sum of positive terms (even permutations) as $p_r$ and sum of negative terms (odd permutations) as $-q_r$.  Then Rutherford expresses the characteristic equation applicable for both rings and semirings as
\begin{align*}
    \lambda^n+q_1 \lambda^{n-1}+p_2 \lambda^{n-2} + q_3 \lambda^{n-3} + \ldots 
    &=  p_1 \lambda^{n-1}+q_2 \lambda^{n-2} + p_3 \lambda^{n-3} + \ldots
\end{align*}
For any $K \subseteq \mathbb{N}$, let $S_K$ be the set of all permutations of $K$, and let $\sgn(\sigma)$ be $+1$ for an even permutation and $-1$ for an odd permutation.  The characteristic polynomial is
\begin{align}
\sum_{K \subseteq [d]} \sum_{\pi \in S_K \atop \sgn(\pi) \neq (-1)^{|K|}} \left(\prod_{i\in K} \mu_{i,\pi(i)}\right) \lambda^{d-|K|}
&= \sum_{K \subseteq [d]} \sum_{\pi \in S_K \atop \sgn(\pi) = (-1)^{|K|}} \left(\prod_{i \in K} \mu_{i,\pi(i)}\right) \lambda^{d-|K|}. \label{eq:charpoly}
\end{align}

If we can ensure that the characteristic equation has just $\lambda^d$ on the left-hand side, then we have a compact representation for inside weights. The following result characterizes the graphs for which this is true.

\begin{theorem}
Given a semiring-weighted directed graph $G$, the characteristic equation of $G$'s adjacency matrix, given by the semiring version of CHT, has only $\lambda^d$ on its left-hand side if and only if $G$ does not have two node-disjoint cycles.
\end{theorem}

\begin{proof}
Let $K$ be a node-induced subgraph of the directed graph $G$.
A \emph{linear subgraph} of $K$ is a subgraph of $K$ that contains all nodes in $K$ and each node has indegree and outdegree 1 within the subgraph, that is, a collection of directed cycles such that each node in $K$ occurs in exactly one cycle. Every permutation $\pi$ of $K$ corresponds to the linear subgraph of $K$ containing edges $(i, \pi(i))$ for each $i \in K$ \citep{harary:1962}. 

Note that $\sgn(\pi) = +1$ iff the corresponding linear subgraph has an even number of even-length cycles.  Moreover, note that $\sgn(\pi)=(-1)^{|K|}$ appearing in (\ref{eq:charpoly}) holds iff the corresponding linear subgraph has an even number of cycles (of any length). So if the transition graph does not have two node-disjoint cycles, the only nonzero term in (\ref{eq:charpoly}) with $\sgn(\pi)=(-1)^{|K|}$ is that for which $K = \emptyset$, that is, $\lambda^d$. To prove the other direction, suppose that the graph does have two node-disjoint cycles; then the linear subgraph containing just these two cycles corresponds to a $\pi$ that makes $\sgn(\pi)=(-1)^{|K|}$. \qed
\end{proof}

The coefficients in (\ref{eq:charpoly}) look difficult to compute; however, the product inside the parentheses is zero unless the permutation $\pi$ corresponds to a cycle in the transition graph of the automaton.  Given that we are interested in computing this product on linear subgraphs, we are only concerned with simple cycles.  Using an algorithm by Johnson \cite{johnson1975finding}, all simple cycles in a directed graph can be found in $\mathcal{O}((n+e)(c+1))$ with $n=\text{number of nodes}$, $e=\text{number of edges}$, and $c=\text{number of simple cycles}$.

\begin{theorem}
    A digraph D with no two disjoint dicycles has at most $2^{|V|-1}$ simple dicycles.
\end{theorem}

\begin{proof}
First, a theorem from Thomassen \cite{thomassen1987digraphs} limits the number of cases we must consider.  In the first case, one vertex, $v_s$, is contained in every cycle.  If we consider $G \setminus \{v_s\}$, this is a directed acyclic graph (DAG) and thus there is a partial order determined by reachability.  This partial order determines the order that vertices appear in any cycle in $G$, which limits the number of simple cycles to the number of choices for picking vertices to join $v_s$ in each cycle.  This is a binary choice on $|V|-1$ vertices, thus $2^{|V|-1}$ possible cycles
(see Figure~\ref{fig:simpleCycleBound}).

In the second case, the graph contains a subgraph with $3$ vertices with no self loops, but all $6$ other possible edges between them.  If we let $S$ be the set of these three vertices, then $G \setminus S$ has a partial order on it just as in the first case.  Additionally, for each $s \in S$, there exists a partial order on $G \setminus (S\setminus \{s\})$, and these uniquely determine the order of vertices in any cycle in $G$.  While the bound could be lowered, this is bounded above by $2^{|V|-1}$.

All other cases can be combined with the second case by observing that they all start with the same graph as the second case, then modified by subdivision (breaking an edge in two by inserting a vertex in the middle) or splitting (breaking a vertex in two, one with all in edges, one with all out edges, then adding one edge from the in vertex to the out vertex).  These cases do not violate the arguments of the second case, nor add any additional cycles.  Intuitively, these are graphs from case two with some edge(s) deleted. \qed
\end{proof}

\begin{figure}
\begin{center}
\begin{tikzpicture}

\node[state] (s) at  (0,0) {};
\node[state] (a) at  (0.75,0) {};
\node[state] (b) at (1.5,0) {};
\node[state] (c) at (2.25,0) {};
\node[state] (n) at (4.5,0) {};
\draw (s) edge[loop left] (s);

\draw (s) edge[bend right] (a);
\draw (s) edge[bend right] (b);
\draw (s) edge[bend right] (c);
\draw (s) edge[bend right] (n);

\draw (b) edge[bend right] (s);
\draw (c) edge[bend right] (s);
\draw (n) edge[bend right] (s);

\draw (a) edge (s);
\draw (b) edge (a);
\draw (c) edge (b);

\draw (c) edge[bend right] (a);
\draw (n) edge[bend right] (a);
\draw (n) edge[bend right] (b);
\draw (n) edge[bend right] (c);

\path (c) to node {\dots} (n);
\node (a3) at (3,0) {};
\draw (a3) to (c);
\node (a3) at (3.75,0) {};
\draw (n) to (a3);

\end{tikzpicture}
\end{center}
\caption{A directed graph achieving the $2^{|V|-1}$ simple cycle bound.}
\label{fig:simpleCycleBound}
\end{figure}
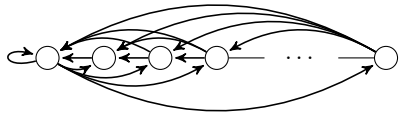

\subsection{Digression: Binary languages and beyond}
\label{sec:binary}

If $\Sigma$ has two symbols and the transition matrices are commuting matrices over a field, then inside weights can still be represented in $d$ dimensions \cite{gerstenhaber}. We give only a brief sketch here of the simpler, algebraically closed case \cite{barria+halmos}.

Given a matrix $M$ with entries in an algebraically closed field, there exists a matrix $S$ such that $S^{-1}MS$ is in \emph{Jordan form}.  A matrix in Jordan form has the following block structure.  Each $A_i$ is a square matrix and $\lambda_i$ is an eigenvalue.

\begin{equation*}
S^{-1}MS=
  \begin{bmatrix}
    A_1&  & 0\\
     & \ddots &  \\    
    0 &  & A_p\\
  \end{bmatrix}
\qquad
A_i = 
  \begin{bmatrix}
    \lambda_i & 1 &  & 0\\
     & \lambda_i & \ddots & \\
     &  & \ddots & 1 \\
    0 &  &  & \lambda_i
  \end{bmatrix}
\end{equation*}
Let the number of rows in $A_i$ be $k_i$.  Here let $M = \mu(a)$ be one of the commuting transition matrices.  Then the following matrices span the algebra generated by the commuting transition matrices $\mu(a)$ and $\mu(b)$:
\begin{align*}
    1, \mu(a), & \ldots, \mu(a)^{k_1-1},\\
    \mu(b), \mu(a)\mu(b), & \ldots, \mu(a)^{k_2-1}\mu(b),\\
    &\vdotswithin{\ldots}\\
    \mu(b)^{p-1}, \mu(a)\mu(b)^{p-1}, &\ldots, \mu(a)^{k_p-1}\mu(b)^{p-1}.
\end{align*}
The number of matrices in this span is equal to the dimension of $\mu(a)$ and $\mu(b)$, which in our case is $d$.  Further, a basis for the algebra is contained within this span.  Therefore the inside weights can be represented in $d$ dimensions.

On the other hand, if the weights come from a ring, the above fact does not hold in general \cite{holbrook+omeara}.
Going beyond binary languages, if $\Sigma$ has four or more symbols, then inside weights might need as many as $\lfloor d^2/4 \rfloor+1$ dimensions, which is not much of an improvement \cite{gerstenhaber}
. The case of three symbols remains open \cite{holbrook+omeara}.

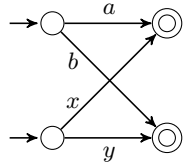
\begin{figure}
\begin{center}
\begin{tikzpicture}[x=0.6in,y=0.6in]
\node[initial,state](q00) at (0,0) {};
\node[initial,state](q01) at (0,1) {};
\node[accepting,state](q10) at (1,0) {};
\node[accepting,state](q11) at (1,1) {};
\draw (q00) edge[below] node{$y$} (q10);
\draw (q00) edge[very near start,above] node{$x$} (q11);
\draw (q01) edge[very near start,below] node{$b$} (q10);
\draw (q01) edge[above] node{$a$} (q11);
\end{tikzpicture}
\end{center}
\caption{Example commutative automaton whose inside weights require storing more than $d$ values.}
\label{fig:gerstenhaber}
\end{figure}

\subsection{Regular expressions}
\label{sec:combination}

Based on the above results, we might not be optimistic about efficiently representing inside weights for languages other than unary languages. But in this subsection, we show that for multiset automata converted from multiset regular expressions, we can still represent inside weights using only $d$ coefficients. We show this inductively on the structure of the regular expression.

First, we need some properties of the matrices $\kappa(a)$.
\begin{lemma} \label{lem:kappa}
If $\mu(a)$ and $\kappa(a)$ are constructed from a multiset regular expression, then
\begin{enumerate}
\item $\kappa(a)\kappa(a) = \kappa(a)$. \label{item:idempotent}
\item $\kappa(a)\kappa(b) = \kappa(b)\kappa(a)$.
\item $\mu(a) \kappa(a) = 0$.
\item $\mu(a) \kappa(b) = \kappa(b) \mu(a)$ if $a \neq b$.
\end{enumerate}
\end{lemma}

To show that $\ins(M)$ 
can be expressed in 
$d$ dimensions, we will need to prove an additional property about the structure of $\ins(M)$.  
Note that if $\ins(M)$ is not a free-module, then $\dim\ins(M)$ is the size of the generating set we construct.
\begin{theorem}
If $M$ is a ring-weighted multiset automaton with $d$ states converted from a regular expression, then \begin{enumerate}
\item $\dim\ins(M) = d$.
\item\label{item:decompose} $\ins(M)$ can be decomposed into a direct sum \[ \ins(M) \cong \bigoplus_{\Delta \subseteq \Sigma} \ins_\Delta(M) \]
where $\mu(w) \in \ins_\Delta(M)$ iff $\alphabet(w) = \Delta$.
\end{enumerate}
\end{theorem}

\begin{proof}
By induction on the structure of the regular expression $\alpha$.

If $\alpha$ is unary: the Cayley-Hamilton theorem gives a 
generating set \\
$\{I, \mu(a), \ldots, \mu(a)^{d-1}\}$, which has size $d$. Moreover, let $\ins_\emptyset(M)$ be the span of $\{I\}$ and $\ins_{\{a\}}(M)$ be the span of the $\mu(a)^i$ ($i>0$). The automaton $M$, by construction, has a state (the initial state) with no incoming transitions. That is, its transition matrix has a zero column, which means that its characteristic polynomial has no $I$ term. Therefore, if $w \neq \epsilon$, $\mu(w) \in \ins_{\{a\}}(M)$.

If $\alpha = k\alpha_1$, then $\ins(M) = \ins(M_1)$, so both properties hold of $\ins(M)$ if they hold of $\ins(M_1)$.

If $\alpha = \alpha_1 \cup \alpha_2$, the inside weights of $M_1 \cup M_2$ for $w$ are
\[\mu(w) = \prod_{a \in w} \mu(a) = \prod_{a} \begin{bmatrix} \mu_1(a) & 0 \\ 0 & \mu_2(a) \end{bmatrix} = \begin{bmatrix} \prod_a \mu_1(a) & 0 \\ 0 & \prod_a \mu_2(a) \end{bmatrix} = \begin{bmatrix} \mu_1(w) & 0 \\ 0 & \mu_2(w) \end{bmatrix}.\] 
Thus, $\ins(M) \cong \ins(M_1) \oplus \ins(M_2)$, and $\dim \ins(M) = \dim \ins(M_1) + \dim \ins(M_2)$. Moreover, $\ins_\Delta(M) \cong \ins_\Delta(M_1) \oplus \ins_\Delta(M_2)$.

If $\alpha = \alpha_1 \alpha_2$, the inside weights of $M_1 \shuffle M_2$ for $w$ are
\begin{align*}
\mu(w) &= \prod_{a \in w} \mu(a) = \prod_{a \in w} (\mu_1(a) \otimes \kappa_2(a) + I \otimes \mu_2(a))  \\
&= \sum_{uv=w} \left( \prod_{a \in u} \mu_1(a) \otimes \prod_{a \in u} \kappa_2(a) \prod_{a \in v} \mu_2(a) \right) \\
&= \sum_{uv=w} \mu_1(u) \otimes \kappa_2(u) \mu_2(v)
\end{align*}
where we have used Lemma~\ref{lem:kappa} and properties of the Kronecker product. Let $\{e_i\}$ and $\{f_i\}$ be a 
generating set
for $\ins(M_1)$ and $\ins(M_2)$, respectively. Then the above can be written as a linear combination of terms of the form $e_i \otimes \kappa_2(u) f_j$. We take these as 
a generating set 
for $\ins(M)$. Although it may seem that there are too many 
generators, 
note that if both $\mu_1(u)$ and $\mu_1(u')$ depend on $e_i$, they belong to the same  
submodule
and therefore use the same symbols, so $\kappa_2(u) = \kappa_2(u')$ (Lemma~\ref{lem:kappa}.\ref{item:idempotent}). Therefore, the $e_i \otimes \kappa_2(u) f_j$ form a 
generating set 
of size $\dim\ins(M_1) \cdot \dim\ins(M_2)$.

Moreover, let $\ins_\Delta(M)$ be the 
submodule 
spanned by all the $\mu_1(u) \otimes \kappa_2(u) \mu_2(v)$ such that $\alphabet(uv) = \Delta$. \qed
\end{proof}

\section{Conclusion}
We have examined weighted multiset automata, showing how to construct them from weighted regular expressions, how to learn weights automatically from data, and how, in certain cases, inside weights can be computed more efficiently in terms of both time and space complexity.  We leave implementation and application of these methods for future work.

\section*{Acknowledgements}

We would like to thank the anonymous reviewers for their very detailed and helpful comments.

This research is based upon work supported by the Office of the Director of
National Intelligence (ODNI), Intelligence Advanced Research Projects
Activity (IARPA), via AFRL Contract \#FA8650-17-C-9116.
The views and conclusions contained herein are those of the authors and
should not be interpreted as necessarily representing the official policies or
endorsements, either expressed or implied, of the ODNI, IARPA, or the
U.S. Government. The U.S. Government is authorized to reproduce and
distribute reprints for Governmental purposes notwithstanding any
copyright annotation thereon.

\bibliography{references}

\end{document}